\documentclass[12pt,a4paper]{article}       

\usepackage{amsthm}
\usepackage{amsmath}
\usepackage{graphicx}
\usepackage{amssymb}
\usepackage{fullpage,color}

\newtheorem{thm}{Theorem}[section]
\newtheorem{lma}[thm]{Lemma}
\newtheorem{cor}[thm]{Corollary}
\newtheorem{conj}[thm]{Conjecture}

\newtheorem{obs}[thm]{Observation}

\newtheorem{rmk}[thm]{Remark}

\DeclareMathOperator{\xsc}{\chi_{\textrm{SC}}}
\DeclareMathOperator{\isc}{\chi_{\textrm{ISC}}}

\begin{document}

\title{The Interactive Sum Choice Number of Graphs\footnote{An extended abstract of this paper was accepted to Eurocomb'17.}}
\author{Marthe Bonamy\footnote{CNRS, LaBRI, Universit\'e de Bordeaux, France. Email: marthe.bonamy@u-bordeaux.fr} \and Kitty Meeks\footnote{School of Computing Science, University of Glasgow, United Kingdom. Email: kitty.meeks@glasgow.ac.uk.}}

\date{}

\maketitle

\parindent=0pt
\parskip=6pt

\begin{abstract}
We introduce a variant of the well-studied sum choice number of graphs, which we call the \emph{interactive sum choice number}.  In this variant, we request colours to be added to the vertices' colour-lists one at a time, and so we are able to make use of information about the colours assigned so far to determine our future choices.  The interactive sum choice number cannot exceed the sum choice number and we conjecture that, except in the case of complete graphs, the interactive sum choice number is always strictly smaller than the sum choice number.  In this paper we provide evidence in support of this conjecture, demonstrating that it holds for a number of graph classes, and indeed that in many cases the difference between the two quantities grows as a linear function of the number of vertices.
\end{abstract}

\section{Introduction}

The \emph{choice number} of a graph $G$ is the minimum length of colour-list that must be assigned to each vertex of $G$ so that, regardless of the choice of colours in these lists, there is certain to be a proper colouring of $G$ in which every vertex is coloured with a colour from its list.  A small subgraph of $G$ which is, in some sense, ``hard'' to colour, can therefore force the choice number for $G$ to be large, even if most of the graph is ``easy'' to colour.  The \emph{sum choice number} of $G$ (written $\xsc(G)$), introduced by Isaak \cite{isaak02}, captures the ``average difficulty'' of colouring a graph: each vertex can now be assigned a different length of colour-list, and the aim is to minimise the sum of the list lengths (while still guaranteeing that there will be a proper list colouring for any choice of lists). A long odd cycle is an example of a graph where most of the graph is “easier” to colour than the choice number indicates. 

For any graph $G = (V,E)$, we have $\xsc(G) \leq |V| + |E|$: we can order the vertices arbitrarily and assign to each vertex $d^-(v)+1$ colors, where $d^-(v)$ is the number of neighbours of $v$ that are before it in the order, and colour greedily in this order.  Graphs for which this so-called greedy bound is in fact equal to the sum choice number are said to be \emph{sc-greedy}, and one of the main topics for research into the sum choice number has been the identification of families of graphs which are (or are not) sc-greedy; we discuss known results about the sum choice number in more detail in Section \ref{sec:sum-choice}. 

In this paper we introduce a variation of the sum choice number, called the \emph{interactive sum choice number} of $G$, in which we do not have to determine in advance all of the lengths of the colour lists: at each step we ask for a new colour to be added to the colour list for some vertex of our choosing and, depending on what colours have been added to lists so far, we can adapt our strategy.  Formally, we define this quantity in terms of a game played by Alice and Bob, which takes a graph $G$ as input.  Initially, each vertex in $G$ has an empty colour-list; at each round, Alice chooses a vertex $v$ and Bob must add a colour (that does not already appear in the list) to the list at $v$.  The game terminates when $G$ admits a proper list colouring; Alice seeks to minimise the number of rounds, while Bob aims to maximise this quantity.  The interactive sum choice number of $G$, written $\isc(G)$, is the number of rounds when both Alice and Bob play optimally on $G$.

It is clear that $\isc(G) \leq \xsc(G)$ for any graph $G$, as Alice can simply ask for the appropriate number of colours to be added to the list for each vertex without paying any attention to the colours that have been added so far.  The natural question is then whether we are in fact able to improve on the sum choice number of $G$ by exploiting partial information about the colour lists.  A simple example where we can improve on the sum choice number is the three-vertex path $P_3$, whose sum choice number is known to be $5$.  In the interactive setting, Alice can start by requesting one colour for each vertex, and then adapt her strategy based on these initial colours: if the two endpoints are given the same first colour, she can obtain a proper colouring by requesting just one more colour for the middle vertex; otherwise, at most one of the endpoints, say $v$, has the same colour as the middle vertex, and so she obtains a proper colouring by requesting an additional colour for $v$.  Thus she can always obtain a proper colouring with at most four requests, a strict improvement on the sum choice number.

However, interactivity does not always allow us to improve on the sum choice number.  For a counterexample, let $G = (V,E)$ be a complete graph, and suppose that for every vertex $v \in V$, the first time Alice asks to add a colour to the colour list for $v$ it will be given colour $1$, the second time it will be given colour $2$, and so on.  Then, whatever order she requests to add colours, we know that there can then be at most one vertex for which we never request a second colour (otherwise two adjacent vertices would have to be assigned colour $1$), at most one vertex for which we request exactly two colours, and more generally for each $1 \leq i \leq n = |V|$ there can be at most $1$ vertex for which we request exactly $i$ colours in total.  Thus we see that 
$$\isc(G) \geq \sum_{i=1}^n i = |V| + |E| \geq \xsc(G).$$
The same argument can easily be extended to graphs that are the disjoint union of complete graphs.

However, we believe that this may be the only case in which there is equality:

\begin{conj}\label{conj:all-non-complete}
If any connected component of $G$ is not a complete graph, then $\isc(G) < \xsc(G)$.
\end{conj}

The main purpose of this paper is to give evidence for Conjecture~\ref{conj:all-non-complete}. We confirm it for sc-greedy graphs, and prove more strongly that the gap between $\isc(G)$ and $\xsc(G)$ is an increasing function of the number of vertices for various graph classes including trees, unbalanced complete bipartite graphs and grids (the latter two being classes which are known not to be sc-greedy).

Two other variants of sum-choosability have also been introduced recently.  In the \emph{sum-paintability} variant \cite{carraher15,mahoney15}, the \emph{painter} decides a budget for each vertex in advance (as in sum list colouring), then in each round the \emph{lister} reveals a subset of vertices which have colour $c$ in the list and the painter must decide immediately which of these vertices to paint with colour $c$.  Thus, there is less information available than in the standard setting of sum-choosability since painter must fix the colour of some vertices before knowing the entire contents of the colour lists.  The relationship between the interactive sum choice number and the second of these variants, the \emph{slow-colouring game} \cite{mahoney15slowcol,puleo19} is less clear.  In this variant, at each round, lister reveals a nonempty subset $M$ of the remaining vertices (scoring $|M|$ points), from which painter chooses an stable set to delete; painter seeks to maximise the total score when all vertices have been deleted, while lister seeks to minimise this quantity.  Compared with our setting, lister has the advantage of discovering at the same time all vertices which are permitted to use colour $c$, but on the other hand he must decide immediately which of these to colour with $c$.  In the special case of trees, however, Puleo and West have demonstrated that the same number of rounds are required in both games \cite{puleo19}.

In the remainder of this section, we first give formal definitions of both the sum choice number and interactive sum choice number in Section \ref{sec:defs}, then provide some background and useful results about the sum choice number in Section \ref{sec:sum-choice}.  In Section \ref{sec:basic-facts} we derive some basic properties of the interactive sum choice number, before providing upper bounds on the interactive sum choice number (and showing that these bounds improve on the sum choice number) for various families of graphs in Sections \ref{sec:sc-greedy}, \ref{sec:not-greedy} and \ref{sec:maybe-greedy}.

\subsection{Definitions and notation}
\label{sec:defs}

For general graph notation we refer the reader to \cite{Diestel}. Throughout this paper, we only consider connected graphs.

A proper colouring of a graph $G=(V,E)$ is a function $c:V \rightarrow \mathbb{N}$ such that for every $uv \in E$ we have $c(u) \neq c(v)$.  Given a list assignment $L:V(G) \rightarrow \mathcal{P}(\mathbb{N})$, an $L$-colouring is a proper colouring $c$ such that $c(v) \in L(v)$ for every $v \in V$.  The choice number of a graph $G=(V,E)$ is the smallest natural number $k$ such that $G$ is $L$-colourable for any list assignment $L$ with $|L(v)| \geq k$ for all $v \in V$.  A graph is said to be $k$-choosable if its choice number is at most $k$.

\subsubsection*{Sum choice number}

A function $f:V(G) \rightarrow \mathbb{N}$ is a \emph{choice function} for $G$ if $G$ is $L$-colourable for any list assignment $L : V(G) \rightarrow \mathcal{P}(\mathbb{N})$ such that $|L(v)| \geq f(v)$ for every vertex $v \in V(G)$. The \emph{sum choice number} of $G$, written $\xsc(G)$, is the minimum sum of values of a choice function, namely:
$$\xsc(G)=\min \big\{  \sum_{v \in V} f(v) \ | \ f \text{ choice function for }G \big\}.$$

\subsubsection*{Interactive sum choice number}

Recall that the interactive sum choice number is defined formally as a game, whose input is a graph $G=(V,E)$.  Initially, every vertex $v \in V$ has an empty colour-list, $L_v$. We set $L: v \mapsto L_v$.

At each round, Alice chooses a vertex $v$, and Bob must add to $L_v$ a colour that does not already belong to $L_v$. The game terminates when $G$ is $L$-colourable. Alice seeks to minimise the number of rounds, while Bob seeks to maximise this.

The \emph{interactive sum choice number} of $G$ is defined to be the number of rounds before the game terminates, when both players play optimally.  We write $\isc(G)$ for the interactive sum choice number of $G$.

We say a graph $G$ \emph{admits a strategy} of length $k$ if Alice can play in a certain way so that she can always provide a proper $L$-colouring of the graph at the end of round $k$ or before. The \emph{trace} of a strategy on a given graph is the sequence of vertices on which Alice requested a new colour. Note that a trace does not characterize the strategy used, as a strategy can produce many different traces depending on how Bob plays.

\subsection{Background on the sum choice number}
\label{sec:sum-choice}

A lot of work has been done on the sum choice number of graphs, and in particular on determining which graphs are sc-greedy, but relatively little is known.  A particular challenge in proving special cases of our conjecture is that, for many graphs $G$, $\xsc(G)$ is only lower and upper bounded, not fully determined.  We do not attempt to describe the state of the art in research into the sum choice number, but in the remainder of this section we list facts about the sum choice number which we will use in the rest of the paper.

Most of the graphs for which the sum choice number is known exactly are those which have been shown to be sc-greedy.  These include complete graphs \cite{isaak04}, trees \cite{isaak04}, cycles \cite{berliner06}, cycles with pendant paths \cite{heinold-thesis}, the Petersen graph \cite{heinold-thesis}, $P_2 \square P_n$ \cite{heinold-thesis}, generalised theta-graphs $\Theta_{k_1,k_2,k_3}$ (unless $k_1 = k_2 = 1$ and $k_3$ is odd), certain wheels \cite{wheels}, and trees of cycles \cite{lastrina13,heinold-thesis}.

The smallest graph which is \emph{not} sc-greedy is $K_{2,3}$: the greedy bound tells us that $\xsc(K_{2,3})$ is at most 11, but in fact $K_{2,3}$ is 2-choosable, implying that $\xsc(K_{2,3}) \leq 10$ (and it is straightforward to check that in fact $\xsc(K_{2,3}) = 10$).  Another graph which is not sc-greedy but whose sum choice number has been determined exactly is the $3 \times n$ grid, $P_3 \square P_n$:
\begin{thm}[\cite{heinold12}]\label{thm:p3pn}
$\xsc(P_3 \square P_n) = 8n - 3 - \lfloor \frac{n}{3} \rfloor$
\end{thm}

While the sum choice number is not known exactly for most complete bipartite graphs, new bounds have recently been derived for the sum choice number of $K_{p,q}$ in the case that $p$ is much smaller than $q$.

\begin{thm}[\cite{furedi16}]
There exist positive constants $c_1$ and $c_2$ such that, for all $p \geq 2$ and $q \geq 4p^2\log p$,
$$2q + c_1 p \sqrt{q \log p } \leq \xsc(K_{p,q}) \leq 2q + c_2 p \sqrt{q \log p}.$$
\end{thm}

We finish this section by noting two simple facts about the sum choice number of graphs.  First of all, if $H$ is a subgraph of $G$, then $\xsc(H) \leq \xsc(G)$.  Secondly, if $G$ is the disjoint union of two graphs $G_1$ and $G_2$ then $\xsc(G) = \xsc(G_1) + \xsc(G_2)$.

\section{Basic facts about interactive sum list colouring}
\label{sec:basic-facts}

In this section we will discuss a number of simple facts about interactive sum list colouring which we will exploit throughout the rest of the paper.  These include assumptions we can make about Alice's strategy, ways to modify Alice's strategy, and bounds on the interactive sum choice number of graphs with specific properties.

The first and surprisingly useful observation is that Alice will necessarily need to request a colour for each vertex at some point along the strategy, and thus can make these requests consecutively at the start of the game without increasing the number of moves required.

\begin{obs}\label{obs:oneeach}
Given a graph $G$ on $n$ vertices $v_1,\ldots,v_n$, Alice has a strategy of length $k$ for $G$ iff she has a strategy of length $k$ for $G$ such that any trace starts with $(v_1,v_2,\ldots,v_n)$.
\end{obs}

Therefore, we introduce the notion of $\alpha$-strategy, for $\alpha : V(G) \rightarrow \mathbb{N}$. An $\alpha$-strategy is a strategy assuming $L$ starts as $(v \mapsto [\alpha(v)])_{v \in V(G)}$ instead of $(v \mapsto [\ ])_{v \in V(G)}$.

\begin{obs}\label{obs:firstchoices}
Given a graph $G$ on $n$ vertices, Alice has a strategy of length $k$ for $G$ iff for every function $\alpha : V(G) \rightarrow \mathbb{N}$ she has an $\alpha$-strategy of length $k-n$.
\end{obs}

We can also combine strategies, as follows.

\begin{obs}\label{obs:combination}
Given a graph $G$ on $n$ vertices, and $\alpha : V(G) \rightarrow \mathbb{N}$, if within $k_1$ requests, there is a function $\beta : V(G) \rightarrow \mathbb{N}$ such that $\beta(v) \in L(v)$ for every $v \in V(G)$, and Alice has a $\beta$-strategy of length $k_2$, then Alice has an $\alpha$-strategy of length $k_1+k_2$.
\end{obs}

Observation~\ref{obs:combination} is particularly helpful when dealing with graphs which admit a small vertex or edge cut.

We now describe how Alice can modify her strategy to ensure that a particular colour is not used to colour a given vertex.

\begin{lma}\label{lma:notc}
Given a graph $G$ on $n$ vertices, a vertex $u$ and a colour $c$, if Alice has a strategy of length $k$ for $G$ then Alice has a strategy of length $k+1$ for $G$ such that $G$ admits an $L$-colouring $\beta$ where $\beta(u) \neq c$.
\end{lma}

\begin{proof}
We let Alice unfold her strategy on $G$. Whenever she requests a colour for $u$, we check whether Bob gives her colour $c$. When he doesn't, we keep going. If he does, we request another colour for $u$, and ask Alice to pretend that was the colour Bob gave her in the first place. Alice's strategy terminates in at most $k$ rounds regardless of which colours Bob gives to her, and Bob can only offer $c$ for $u$ once, so the whole strategy terminates in at most $k+1$ rounds.
\end{proof}

Applying this result repeatedly gives the following immediate corollary.

\begin{cor}\label{cor:allnotc}
Given a graph $G= (V,E)$, a set $U \subseteq V$ and a function $f: U \rightarrow \mathcal{P}(\mathbb{N})$, if Alice has a strategy of length $k$ for $G$ then Alice has a strategy of length $k + \sum_{u \in U} |f(u)|$ for $G$ such that $G$ admits an $L$-colouring $\beta$ where, for every $u \in U$, $\beta(u) \notin f(u)$.
\end{cor}

In the remainder of this section we give upper bounds on $\isc(G)$ that depend on properties of $G$.  We begin by considering the size of the largest stable set in $G$.

\begin{lma}\label{lma:stable}
If $\alpha$ is the number of vertices in the largest stable set in $G$, then $\isc(G) \geq 2|V(G)| - \alpha$.
\end{lma}
\begin{proof}
Suppose, for a contradiction, that Alice has a strategy that will guarantee a proper colouring after only $2|V(G)|- \alpha - 1$ rounds.  This must work whatever colours Bob chooses, so we may assume that every vertex receives colour $1$ as its first colour.  Once every vertex has been given its first colour, we have $|V(G)|- \alpha - 1$ moves remaining.  There is therefore a set $W$ of size at least $\alpha + 1$ vertices in which none is given a second colour; however, since the largest stable set in $G$ contains $\alpha$ vertices, there exist two adjacent vertices in $W$.  These vertices must both end up with colour $1$, contradicting the assumption that Alice can find a proper colouring.
\end{proof}

In the next two lemmas, we consider the relationship between the interactive sum choice number of a graph and that of its subgraphs.

\begin{lma}
Let $G$ be a graph and $H$ a subgraph of $G$.  Then $\isc(H) \leq \isc(G) - |V(G) \setminus V(H)|$.
\label{lma:del-edges}
\end{lma}
\begin{proof}
Suppose Alice has a strategy which will guarantee a proper colouring of $G$ after $k$ moves.  She can play this same strategy on $H$, simply omitting any rounds in which she would request a colour for a vertex in $V(G) \setminus V(H)$.  This will certainly give a proper colouring of $H$.  Moreover, Alice's strategy to colour $G$ must include at least one round in which she requests a colour for each vertex in $V(G) \setminus V(H)$, so in her new strategy she omits at least $|V(G) \setminus V(H)|$ rounds.  Thus she is guaranteed to obtain a proper colouring of $H$ after at most $k - |V(G) \setminus V(H)|$ rounds.
\end{proof}

\begin{lma}
Let $G$ be a graph, and let $H$ be an induced subgraph of $G$.  Then
$$\isc(G) \leq \isc(H) + \isc(G[V(G) \setminus V(H)]) + |E(V(H),V(G)\setminus V(H)|.$$
\label{lma:del-subg}
\end{lma}
\begin{proof}
Alice begins by applying a strategy of length $\isc(G[V(G) \setminus V(H)])$ to obtain a proper colouring $\beta$ of $G[V(G)\setminus V(H)]$.  The goal is now to obtain a proper colouring $\gamma$ of $H$ so that combining $\beta$ and $\gamma$ gives a proper colouring of $G$.  Note that $\gamma$ can be any proper colouring of $H$ that satisfies the following additional condition: for any edge $uv$ with $u \in V(H)$ and $v \notin V(H)$, $\gamma(u) \neq \beta(v)$.  But we know from Corollary \ref{cor:allnotc} that Alice has a strategy to obtain such a colouring $\gamma$ for $H$ in at most $\isc(H) + |E(V(H),V(G)\setminus V(H)|$ rounds.  She can therefore now apply this strategy (after obtaining her initial proper colouring $\beta$) to obtain a proper colouring of $G$ after a total of at most $\isc(H) + \isc(G[V(G) \setminus V(H)]) + |E(V(H),V(G)\setminus V(H)|$ rounds, as required.
\end{proof}

We now make a simple observation about the disjoint union of two graphs.

\begin{rmk}
Let $G$ be the disjoint union of two graphs $G_1$ and $G_2$.  Then
$$\isc(G) = \isc(G_1) + \isc(G_2).$$
\label{rmk:disj-union}
\end{rmk}
\begin{proof}
Alice first applies her strategy to obtain a proper colouring of $G_1$, and then applies her strategy to obtain a proper colouring of $G_2$.
\end{proof}

Surprisingly, removing a single edge can make a relatively big difference. Recall that $\isc(K_p)=\sum_{i=1}^{p}i=\frac{p \cdot (p+1)}2$.

\begin{thm}\label{thm:cliques}
For every $p\geq 2$, if $e$ is any edge of the complete graph $K_p$, we have $\isc(K_p - e) \leq \frac{p \cdot (p+1)}2 - \frac{p-2}3$.
\end{thm}
\begin{proof}
Let $G$ be a graph isomorphic to $K_p-e$, and let $u$ and $v$ be the only non-adjacent vertices in $G$. Fix a function $\alpha : V(K_p-e) \rightarrow \mathbb{N}$. Alice requests $\frac{p-2}3$ extra colours on each of $u$ and $v$. We consider two cases depending on whether $L(u) \cap L(v) = \emptyset$. All along this, when we \emph{deal} with a vertex $x$, it means Alice requests as many extra colours as needed until there is a colour $c$ available for $x$ that does not appear on any coloured neighbour of $x$; Alice then colours $x$ with $c$.

\begin{itemize}
\item Assume that $L(u) \cap L(v) \neq \emptyset$. Alice colours $u$ and $v$ with the same colour $c$, and we deal with all the other vertices in an arbitrary order. In total, this $\alpha$-strategy involves at most $\frac{(p-2)\cdot (p-1)}2+2 \frac{p-2}3$ requests.
\item Assume now that $L(u) \cap L(v) = \emptyset$. Note that together, $u$ and $v$ have at least $|L(u)\cup L(v)|=2 \frac{p-2}3+2$ colours available. 

We deal with uncoloured vertices other than $u$ and $v$ until one of $u$ and $v$, say $u$, has only one colour available left (or all other vertices are coloured). Then we deal with $u$, and note that this does not require any extra request. We keep dealing with uncoloured vertices other than $v$ until $v$ has only one colour available left (or all other vertices are coloured): we deal with $v$ and keep going until every vertex is coloured. Consider the order $\mathcal{O}$ with which we dealt with vertices. We note that dealing with a vertex requires fewer than its rank in $\mathcal{O}$ extra requests, except for $u$ and $v$ which required no extra request.

Since $|L(u)|=|L(v)|=\frac{p-2}3+1$ and we only deal with $u$ and $v$ when they only have one colour available left or all other vertices are coloured, we have that the rank of each of $u$ and $v$ in $\mathcal{O}$ is at least $\frac{p-2}3+1$. Additionally, since $|L(u)\cup L(v)|=2\frac{p-2}3+2$, the last vertex of $u$ and $v$ that is dealt with is only considered after at least $2 \frac{p-2}3$ other vertices in $G$. Therefore, the sum of orders of $u$ and $v$ in $\mathcal{O}$ is at least $p$. In total, this $\alpha$-strategy involved at most $\sum_{i=1}^{p}(i-1)- (p-2)+2 \cdot \frac{p-2}3 = \frac{(p-2)(p-1)}{2} - (p-2)+ 2 \frac{p-2}{3}$ requests.
\end{itemize}

This gives us a strategy of length at most $\max(p+\frac{(p-2)\cdot (p-1)}2+2 \frac{p-2}3,p+\sum_{i=1}^{p}(i-1)- (p-2)+2 \cdot \frac{p-2}3)=\frac{p(p+1)}2-\frac{p-2}3$, hence the conclusion.
\end{proof}

\section{Graphs that are sc-greedy}
\label{sec:sc-greedy}

In this section we consider the interactive sum choice number for classes of graphs that are known to be sc-greedy.  We begin by showing that our conjecture holds for all sc-greedy graphs, before obtaining better bounds on the interactive sum choice number for trees and cycles.

\subsection{The general case of sc-greedy graphs that are not complete}

In this section we show that Conjecture \ref{conj:all-non-complete} holds for all graphs that are sc-greedy.  We begin by determining the interactive sum choice number of the path on 3 vertices.

\begin{lma}\label{lma:P_3}
Let $P_3$ denote the path on $3$ vertices.  Then $\isc(P_3) = 4$.
\end{lma}
\begin{proof}
It follows immediately from Lemma \ref{lma:stable} that $\isc(P_3) \geq 4$, so it remains to prove that the reverse inequality holds.  By Observation \ref{obs:firstchoices} it suffices to demonstrate that for any $\alpha:V(P_3) \rightarrow \mathbb{N}$ there exists an $\alpha$-strategy of length at most $1$ for $P_3$.

We may assume without loss of generality that $\alpha$ is not a proper colouring of $P_3$.  Let us denote by $x,y,z$ the vertices of $P_3$, where $x$ and $z$ are the endpoints.  There are now two cases to consider.  If $\alpha(x) = \alpha(z)$ then Alice can obtain a proper colouring by requesting one more colour for $y$, and we are done.  If not, then without loss of generality we may assume that $\alpha(x) = \alpha(y) \neq \alpha(z)$; in this case Alice can obtain a proper colouring by requesting one more colour for $x$.  Thus we have shown that Alice has an $\alpha$-strategy of length $1$ in either case, completing the proof.
\end{proof}

We can apply the same ideas to obtain an upper bound on the interactive sum choice number of any graph that contains many vertex-disjoint induced copies of $P_3$.

\begin{lma}\label{lma:disjoint_P3s}
Let $G$ be a graph which contains at least $t$ vertex-disjoint induced copies of $P_3$.  Then $\isc(G) \leq |V(G)| + |E(G)| - t$.
\end{lma}
\begin{proof}
By Observation \ref{obs:firstchoices} it suffices to demonstrate that for every $\alpha:V(G) \rightarrow \mathbb{N}$, there exists an $\alpha$-strategy of length at most $|E(G)|- t$ for $G$. 
We proceed by induction on $t$. If $t=0$, the result follows directly from the fact that $\isc(G) \leq \xsc(G) \leq |V(G)| + |E(G)|$. Assume now that $t>0$ and that, for any graph $H$ which contains at least $t - 1$ vertex-disjoint copies of $P_3$, we have $\isc(H) \leq |V(H)| + |E(H)| - (t-1)$.

Let $\alpha:V(G) \rightarrow \mathbb{N}$.  Fix a set $\mathcal{P}$ of $t$ vertex-disjoint induced copies of $P_3$ in $G$, and choose an arbitrary element $P \in \mathcal{P}$.  We will denote the vertices of $P$ by $x, y$ and $z$, where $x$ and $z$ are not adjacent. It suffices to demonstrate that $G$ has an $\alpha$-strategy of length at most $|E(G)| - t$.  We consider two cases depending on whether $\alpha(x)=\alpha(y)=\alpha(z)$.
\begin{itemize}
\item Assume that $\alpha(x) \neq \alpha(y)$ or $\alpha(z) \neq \alpha(y)$. By symmetry, we can assume $\alpha(x) \neq \alpha(y)$.  Note that $G' = G \setminus \{x,y\}$ contains at least $t-1$ vertex-disjoint copies of $P_3$ so, by the inductive hypothesis, 
\begin{align*}
\isc(G') &\leq |E(G')| + |V(G')| - (t-1) \\
         &= |E(G)| - |E(\{x,y\},V(G'))|  - 1 + |V(G)| - 2 - (t-1)\\
         &= |E(G)| - |E(\{x,y\},V(G'))| + |V(G)| - 2 - t.
\end{align*}
Thus, for any $\beta: V(G) \rightarrow \mathbb{N}$, $G'$ admits a $\beta$-strategy of length at most $|E(G)| - |E(\{x,y\},V(G'))| - t$.  By Lemma~\ref{lma:notc}, we obtain that $G'$ has an $\alpha$-strategy of length at most 
$$|E(G)| - |E(\{x,y\},V(G'))| - t + |E(\{x,y\},V(G'))| = |E(G)| - t$$
such that no neighbour of $x$ (resp. $y)$ receives the colour $\alpha(x)$ (resp. $\alpha(y)$). It follows that $G$ has an $\alpha$-strategy of length at most $|E(G)| - t$, as required.
\item Assume now that $\alpha(x)=\alpha(y)=\alpha(z)$. Note that $G' = G \setminus \{x,z\}$ contains at least $t-1$ vertex-disjoint copies of $P_3$ so, by the inductive hypothesis, 
\begin{align*}
\isc(G') &\leq |E(G')| + |V(G')| - (t-1) \\
		 &= |E(G)| - |E(\{x,z\},V(G'))| + |V(G)| - 2 - (t-1).
\end{align*}
Thus, for any $\beta: V(G) \rightarrow \mathbb{N}$, $G'$ admits a $\beta$-strategy of length at most $|E(G)| - |E(\{x,z\},V(G'))| - (t-1)$.  Applying Lemma~\ref{lma:notc} again, we obtain that $G'$ has an $\alpha$-strategy of length at most 
$$|E(G)| - |E(\{x,z\},V(G'))| - (t-1) + |E(\{x,z\},V(G'))| - |N(x) \cap N(z)|$$ 
such that no neighbour of $x$ or $z$ receives the colour $\alpha(x)$.  Note that $|N(x) \cap N(z)| \geq 1$, as $y$ is a common neighbour of $x$ and $z$, so we see once again that $G$ has an $\alpha$-strategy of length at most $|E(G)| - t$, as required.
\end{itemize}
\end{proof}

It is now straightforward to show that our conjecture holds for all sc-greedy graphs.

\begin{cor}
Let $G$ be a graph which is not a disjoint union of complete graphs.  If $G$ is sc-greedy, then $\isc(G) < \xsc(G)$.
\end{cor}
\begin{proof}
By assumption, some connected component $C$ of $G$ is not a complete graph; therefore $C$ contains at least one copy of $P_3$ as an induced subgraph.  It then follows immediately from Lemma \ref{lma:disjoint_P3s} that $\isc(G) \leq |V(G)| + |E(G)| - 1 = \xsc(G) - 1$.
\end{proof}

\subsection{Trees}

In this section we discuss the interactive sum choice number of trees, which differs from the sum choice number by approximately half the number of vertices.

\begin{thm}\label{thm:trees}
Let $T$ be a tree on $n \geq 3$ vertices.  Then $\isc(T) \leq \lfloor \frac{3n}{2} \rfloor$.
\end{thm}
\begin{proof}
By contradiction. Take a minimal counter-example $T$, and consider $n$ its number of vertices. By Observation~\ref{obs:firstchoices}, let $\alpha : V(G) \rightarrow \mathbb{N}$ such that $T$ has no $\alpha$-strategy of length $\lfloor \frac{n}2 \rfloor$. 

Assume there is in $T$ a vertex $u$ adjacent to two leaves $v_1$ and $v_2$. Let $T'$ be $T \setminus \{v_1,v_2\}$. Note that by minimality of $T$, the tree $T'$ admits a strategy of length $\lfloor \frac{3n}2 \rfloor-3$. If $\alpha(v_1)=\alpha(v_2)$, by Lemma~\ref{lma:notc}, $T'$ admits a strategy of length $\lfloor \frac{3n}2 \rfloor -3+1$ so that $u$ is coloured at the end with a colour other than $\alpha(v_1)$. It follows that $T'$ admits an $\alpha$-strategy of length $\lfloor \frac{n}2 \rfloor$ so that $u$ is coloured at the end with a colour other than $\alpha(v_1)$. We apply this strategy on $T'$, which results in a proper colouring of $T$ with an $\alpha$-strategy of length $\lfloor \frac{n}2 \rfloor$, a contradiction. If $\alpha(v_1)\neq \alpha(v_2)$, we apply an $\alpha$-strategy of length $\lfloor \frac{n}2 \rfloor -1$ on $T'$ and consider the resulting colour $\beta$ of $u$. Assume w.l.o.g. that $\beta \neq \alpha(v_1)$. We request a new colour on $v_2$ if needed: this yields an $\alpha$-strategy for $T$ of length $\lfloor \frac{n}2 \rfloor$, a contradiction.

Therefore, no vertex in $T$ is adjacent to two leaves. Consequently, if $T$ contains more than $2$ vertices, there is a vertex $v$ with $N(v)=\{u,w\}$ where $w$ is a leaf. First assume that $\alpha(v) \neq \alpha(w)$. Then we apply an $\alpha$-strategy on $T \setminus \{v,w\}$ such that $u$ is coloured in the end with a colour other than $\alpha(v)$. By Lemma~\ref{lma:notc} and minimality of $T$, there is such a strategy of length $\lfloor \frac{n}2 \rfloor$. This yields a proper colouring of $T$ with a strategy of length $\lfloor \frac{n}2 \rfloor$, a contradiction. Therefore $\alpha(v)=\alpha(w)$. We apply an $\alpha$-strategy of length $\frac{n}2-1$ on $T \setminus \{v,w\}$, and consider the final colour $\beta$ of $u$. If $\beta=\alpha(v)$, we request a new colour for $v$, and note that by colouring $v$ with it, we obtain a proper colouring of $T$, hence we have a strategy of length $\lfloor \frac{n}2 \rfloor$. If $\beta \neq \alpha(v)$, we request a new colour for $w$, and note that by colouring $w$ with it, we obtain a proper colouring of $T$, hence the conclusion.

Thus $T$ has at most $2$ vertices, a contradiction.
\end{proof}

This bound is tight for paths: since the largest stable set in a path $P$ on $n$ vertices has size exactly $\lceil \frac{n}{2} \rceil$, it follows from Lemma \ref{lma:stable} that $\isc(P) \geq 2n - \lceil \frac{n}{2} \rceil = \lfloor \frac{3n}{2} \rfloor$.  However, the following result shows that we can make a significant improvement on Theorem~\ref{thm:trees} in the case of stars.

\begin{lma}\label{lma:stars}
$\isc(K_{1,p}) = p + q+1$, where $q= \max \{q \in \mathbb{N} | \frac{q*(q+1)}2 \leq p \}$.
\end{lma}
\begin{proof}
We first prove the upper-bound, i.e. $K_{1,p}$ admits a strategy of length $p + q+1$. Let $\alpha$ be an assignment of a colour to each vertex. We will describe an $\alpha$-strategy of length $q$. Let $u$ be the vertex of degree $p$, and $v_1,\ldots,v_p$ be the vertices of degree $1$ (note that if $p=1$ then $\isc(K_{1,1})=\isc(P_2)$, which we know to be $3$). 

Let $c_1=\alpha(u)$. Consider the set $S_1$ of vertices $w$ in $\{v_1,\ldots,v_p\}$ satisfying $\alpha(w)=c_1$. If $|S_1|\leq q$, then we request a new colour for each vertex in $S_1$, and obtain a proper colouring in at most $p+q+1$ rounds. Therefore we can assume $|S_1| \geq q+1$.

We request a new colour for $u$, let $c_2$ be the colour Bob gives. Consider the set $S_2$ of vertices $w$ in $\{\ v_1,\ldots,v_p\}$ satisfying $\alpha(w)=c_2$. If $|S_2|\leq q-1$, then we request a new colour for each vertex in $S_2$, and obtain a proper colouring in at most $p+q+1$ rounds. Therefore we can assume $|S_2| \geq q$.

We can iterate until we find an $L$-colouring or reach the end of the $q^{th}$ round. Assume for contradiction that the strategy fails, i.e. we request $q$ new colours for $u$ but none of them allows us to obtain an $L$-colouring. Then we have for every $i$ from $1$ to $q+1$ that $|S_i|\geq q+2-i$. Note that the $S_i$'s are pairwise disjoint, as they correspond to different values taken by $\alpha$. In particular, we have $p \geq \sum_{i=1}^{q+1}(q+2-i)$, thus $p \geq \sum_{i=1}^{q+1} i = \frac{(q+1)\cdot (q+2)}2$, a contradiction to the choice of $q$. Consequently, the strategy does not fail, and we have an $\alpha$-strategy of length $q$, for every $\alpha$, hence $\isc(K_{1,p}) \leq p + q+1$.

Let us now argue that there is no strategy of length $p+q$. To do so, we will exhibit an assignment $\alpha$ of first choices and a strategy for Bob that will not allow for less than $q$ extra requests from Alice. Assign colour $1$ to $u$ and order the $v_i$'s arbitrarily: assign colour $1$ to the first $q$ vertices $v_i$'s, colour $2$ to the following $q-1$ vertices, colour $3$ to the next $q-2$ vertices, etc, colour $q$ to the next vertex, and colour the rest arbitrarily (note that there remains between $0$ and $q$ vertices). Alice has to obtain an $L$-colouring in $q-1$ rounds. 

Bob's strategy is to offer the colours $(2,\ldots,q)$ for $u$ (in that order), and arbitrary colours for other vertices. Assume Alice has a strategy in $q-1$ rounds, and let $0 \leq k \leq q-1$ be the number of new colours for $u$ that strategy requires. Alice has to finish in $q-k-1$ rounds. Note that for every colour $i$ in $1,\ldots,k+1$, there are at least $q-i+1 \geq q-k$ neighbours of colour $i$. Therefore, Alice's only hope of obtaining a proper colouring is to request new colours for at least $q-k$ neighbours of $u$: a contradiction.
\end{proof}

It is therefore tempting to think that we can use Lemma~\ref{lma:stars} to obtain a more refined bound in the case of trees that are not paths. However, we note that any tree $T$ that admits a perfect matching will satisfy $\isc(T) \geq \frac{3n}{2}$, as it suffices to see that there is no $\alpha$-strategy of length less than $\frac{n}2$ if $\alpha$ assigns the same colour to all the vertices in $T$. Note that a star $K_{1,p+1}$ where $p$ edges are subdivided is very far from being a path, yet admits a perfect matching. We can further note that even if the tree does not admit a perfect matching, the same argument shows that there is no strategy of length less than $n+k$ where $k$ is the size of a smallest vertex cover of $T$ (that is, of a subset $S$ of vertices such that $T-S$ induces a stable set). Nevertheless, this is still not the right bound for all trees, as stars admit a vertex cover of size $1$ and yet require much more than $n+1$ rounds in general. Since the dissemination of an earlier version of this paper, the case of trees has been fully resolved by Puleo and West~\cite{puleo19}.

\subsection{Cycles}

In this section we determine exactly the interactive sum choice number of cycles.

\begin{thm}
Let $C_n$ be the cycle on $n \geq 3$ vertices.  Then $\isc(C_n) = \lfloor \frac{3(n+1)}{2} \rfloor$.
\end{thm}
\begin{proof}
We first argue that $\isc(C_n) \leq \lfloor \frac{3(n+1)}{2} \rfloor$ for every $n$.  Let $v$ be an arbitrary vertex of $C_n$, and note that $C_n \setminus \{v\}$ is a path on $n-1$ vertices.  Thus, by Theorem~\ref{thm:trees} we know that $\isc(C_n \setminus \{v\}) \leq \lfloor \frac{3(n-1)}{2} \rfloor$.  It then follows from Lemma \ref{lma:del-subg} that
$$\isc(C_n) \leq \isc(C_n \setminus \{v\}) + 1 + 2 \leq \left \lfloor \frac{3(n-1)}{2} \right \rfloor + 3 = \left \lfloor \frac{3(n+1)}{2} \right \rfloor.$$

To see the reverse inequality, we consider separately the cases for even and odd $n$.  Suppose first that $n$ is odd, and suppose for a contradiction that Alice has a strategy of length at most $\lfloor \frac{3(n+1)}{2} \rfloor - 1 = \lfloor \frac{3n + 1}{2} \rfloor$.  Since Alice's strategy must work whatever colours Bob chooses, we may assume that the first time Alice requests a colour for any vertex she will be given colour $1$, and the second time she requests a colour for any vertex she will be given colour $2$. Note that, following the reasoning of Lemma \ref{lma:stable}, any strategy must request a second colour for every vertex in some set $U$ such that $V(C_n) \setminus U$ is a stable set, implying that $|U| \geq n - \frac{n-1}{2} = \frac{n+1}{2}$.  Thus our strategy can have length at most $\lfloor \frac{3n + 1}{2} \rfloor$ if and only if $|U| = \frac{n+1}{2}$, and we request exactly one colour for every vertex not in $U$ and two colours for every vertex in $U$, so every vertex outside $U$ ends up with colour list $\{1\}$ and every vertex in $U$ with colour list $\{1,2\}$.  Any proper colouring of $C_n$ respecting these colour lists would give a proper 2-colouring of $C_n$, giving contradiction since $n$ is odd.

Now suppose that $n$ is even, and denote the vertices of $C_n$ as $v_1,\ldots,v_n$ (where $v_iv_j$ is an edge if and only if $j-i \equiv 1 \mod{n}$), and that we have an initial colouring $\alpha$ such that $\alpha(v_{2i-1}) = \alpha(v_{2i}) = i$ for each $1 \leq i \leq n/2$.  Any $\alpha$-strategy must then request an additional colour for at least one vertex in $\{v_{2i-1},v_{2i}\}$ for each $i$, so every $\alpha$-strategy has length at least $n/2$; we can only obtain a strategy of length strictly shorter than $\lfloor \frac{3(n+1)}{2} \rfloor$ if there is an $\alpha$-strategy that involves requesting precisely one extra colour for exactly one vertex in each set $\{v_{2i-1},v_{2i}\}$.  Suppose, therefore, that there exists such an $\alpha$-strategy; denote by $U$ the set of vertices for which Alice requests an extra colour, and notice that every vertex in $U$ must ultimately be coloured with its second colour in order to obtain a proper colouring.  

For each $n+1 \leq i \leq \lfloor \frac{3n + 1}{2} \rfloor$, let $U'$ be the set of vertices on which Alice requests a second colour before round $\lfloor \frac{3n + 1}{2} \rfloor$.  Suppose that Alice selects vertex $v$ at round $i$; we define the following response for Bob:
\begin{itemize}
\item if $i < \lfloor \frac{3n + 1}{2} \rfloor$, add colour $\frac{n}{2} + 1$ to $L_v$ (note that colour $\frac{n}{2} + 1 \neq \alpha(w)$ for any $w \in V(C_n)$);
\item if $i = \lfloor \frac{3n + 1}{2} \rfloor$ and $v$ has a neighbour in $U'$, add $\frac{n}{2} + 1$ to $L_v$;
\item if $i = \lfloor \frac{3n + 1}{2} \rfloor$ and $N(v) = \{x,y\}$ with $x,y \notin U'$, add the unique colour in $\{\alpha(x),\alpha(y)\} \setminus \{\alpha(v)\}$ to $L_v$.
\end{itemize}
We claim that, if Bob always responds in this way, there is no proper $L$-colouring at the end of round $\lfloor \frac{3n + 1}{2}$, giving the required contradiction.  To see this, let $w$ be the vertex Alice chooses at round $\lfloor \frac{3n + 1}{2} \rfloor$, and recall that every vertex in $U$, including $w$, must ultimately be coloured with its second colour in any proper colouring.  If $w$ had a neighbour in $U'$ then there are two adjacent vertices in $U$ which are both assigned colour $\frac{n}{2} + 1$ as their second colour, giving a contradiction.  On the other hand, if $w$ had no neighbour in $U'$, then its neighbours will ultimately be coloured with $\alpha(x)$ and $\alpha(y)$ respectively (as each has a colour list of length one), but $L_w = \{\alpha(x),\alpha(y)\}$, again giving a contradiction.

Therefore Bob can always prevent Alice from obtaining a proper colouring after only $n + \frac{n}{2}$ rounds, and it follows that, when $n$ is even, $\isc(C_n) \geq n + \frac{n}{2} + 1  = \lfloor \frac{3(n+1)}{2} \rfloor$.
\end{proof}

\section{Graphs that are not sc-greedy}
\label{sec:not-greedy}

In this section, we extend our earlier results to show that Conjecture \ref{conj:all-non-complete} also holds for certain classes of graphs that are not sc-greedy.  We consider two families of bipartite graphs: complete (unbalanced) bipartite graphs, and grids.

\subsection{Complete bipartite graphs, the unbalanced case}

In this section we generalise Lemma \ref{lma:stars} to show that Conjecture \ref{conj:all-non-complete} holds for any complete bipartite graph in which the sizes of the two vertex classes are very different.  We will make use of the following result, proved by Kemnitz, Marangio and Voigt, regarding graphs containing a cycle; we will later show that the interactive version can beat this bound.

\begin{lma}[\cite{kemnitz17}]\label{lma:non-treesarelosers}
Let $G$ be a connected graph on $n$ vertices. If $G$ is not a tree, then $\xsc(G) \geq 2n$.
\end{lma}

We now derive an upper bound on the interactive sum choice number for complete bipartite graphs.

\begin{thm}\label{thm:unbalancedcomplete}
For any integers $p$ and $q$, we have $\isc(K_{p,q})\leq p +q+ p^2\sqrt{2q}$.
\end{thm}
\begin{proof}
We will generalize the proof of Lemma~\ref{lma:stars} to the case where, roughly, there are $p$ centers instead of just one. Let $r= \max \{r \in \mathbb{N} | \frac{r*(r+1)}2 \leq q \}$. Note that $r \leq \sqrt{2q}$. Let $\alpha$ be an assignment of a colour to each vertex. We will describe an $\alpha$-strategy of length $p^2 \cdot r$. Let $u_1,\ldots,u_p$ be the vertices of degree $q$, and $v_1,\ldots,v_q$ be the vertices of degree $p$.

We proceed by induction on $p+q$. Let $V_1$ be a largest subset of $\{v_1,\ldots,v_q\}$ that has the same image $c$ by $\alpha$. If $|V_1| \leq r$, there are no $r+1$ vertices in $\{v_1,\ldots,v_q\}$ with the same image by $\alpha$. Alice colours every $u_i$ with $\alpha(u_i)$, and requests an extra $p$ colours on every vertex $v_j$ with $\alpha(v_j) \in \{\alpha(u_i) | 1 \leq i \leq p\}$. This sums up to at most $p^2\cdot r$ requests, as required.

If $|V_1| \geq r+1$, Alice colours every vertex in $V_1$ with $c$, and uses Lemma~\ref{lma:notc} to guarantee that no vertex in $\{u_1,\ldots,u_p\}$ ever receives the colour $c$ as an option, making $p$ extra requests.  This ensures that if there is $\beta$-strategy on the remaining vertices of length at most $\ell$ for every $\beta$, then there is an $\alpha$-strategy on the initial graph of length at most $\ell+p$.  Note that if $V_1 = \{v_1,\ldots,v_q\}$ then we are already done, so we may assume that this is not the case and hence that $r + 1 < q$.   Applying the inductive hypothesis to $G' := G \setminus V_1$, we see that for any $\beta: V(G') \rightarrow \mathbb{N}$, Alice has a $\beta$-strategy on $G'$ of length $p^2 \sqrt{2(q-(r+1))}$.  We are therefore done if $\sqrt{2(q - (r+1))} \leq \sqrt{2q} - 1$, which holds if and only if $2 \sqrt{2q} \leq 2r + 3 \iff 8q \leq 4r^2 + 12 r + 9$.  We know by definition of $r$ that $8q \leq 8 \left(\frac{(r+1)(r+2)}{2} - 1\right) = 4r^2 + 12r$, giving the required bound.
\end{proof}

This gives the following immediate corollary.
\begin{cor}\label{cor:Kpq}
If $p \ll q$ then $\isc(K_{p,q}) < \xsc(G)$.
\end{cor}

It remains to understand how the interactive sum choice number behaves in the balanced case: we leave this is as an open question. Theorem~\ref{thm:unbalancedcomplete} can be marginally generalized, as follows.

\begin{thm}\label{thm:completejoin}
Let $G=(A \cup B,E)$ be a graph such that every vertex in $A$ is adjacent to every vertex in $B$. Then $\isc(G)\leq |A| +|B|+ (\max(\Delta(G[A]),\Delta(G[B]))+1)\cdot |A|^2\sqrt{2|B|}$.
\end{thm}
\begin{proof}
We merely apply the same strategy as for Theorem~\ref{thm:unbalancedcomplete}. We simply replace every request for a new colour on some $u_i$ (which correspond here to vertices in $A$) with $\Delta(G[A])+1$ requests, so as to guarantee that no matter the colouring of the rest of $G$ there is always one out of the $(\Delta(G[A])+1)$ that is not used on a neighbour of $u_i$ in $G[A]$. We proceed symmetrically on $B$. This yields an $\alpha$-strategy for $G$ of length at most $(\max(\Delta(G[A]),\Delta(G[B])+1) \cdot (|A|^2 \sqrt{2|B|})$ for every $\alpha$, hence the conclusion.
\end{proof}

By Lemma~\ref{lma:non-treesarelosers} and Lemma~\ref{lma:del-edges}, we obtain in particular the following corollary.

\begin{cor}\label{cor:unbalanced-bipartition}
Let $G=(A \cup B,E)$ be a graph. If $(\max(\Delta(G[A]),\Delta(G[B]))+1)\cdot |A|^2\sqrt{2|B|}< |A|+|B|$, then $\isc(G) < \xsc(G)$.
\end{cor}

Corollary~\ref{cor:unbalanced-bipartition} also implies that fans on at least 19 vertices satisfy Conjecture~\ref{conj:all-non-complete}, by Lemma~\ref{lma:non-treesarelosers}.

\subsection{Grids}

In this section we demonstrate that the interactive sum choice number is strictly smaller than the sum choice number of the $k \times \ell$ grid $G_{k,\ell}$ for any positive integers $k$ and $\ell$; in fact we prove the following result.

\begin{thm}
Let $G_{k,\ell}$ denote the $k \times \ell$ grid, where $k \leq \ell$, and suppose that $\ell \geq 3$.  Then
$$\xsc(G_{k,\ell}) - \isc(G_{k,\ell}) \geq \frac{1}{18}k \ell.$$
\label{thm:grid-diff}
\end{thm}

Note first that the result follows immediately from Theorem~\ref{thm:trees} in the case that $k = 1$, as in this case $G_{k,\ell}$ is in fact a path on $\ell$ vertices.

If $k=2$, then $G_{k,\ell}$ is a tree of cycles and so is known to be sc-greedy \cite{heinold-thesis,lastrina13}; thus $\xsc(G_{k,\ell}) = 5 \ell - 2$.  On the other hand, by regarding $G_{k,\ell}$ as two paths of length $\ell$ with a total of $\ell$ cross-edges between the two paths, we can apply Lemma \ref{lma:del-subg} to see that $\isc(G_{k,\ell}) \leq 2 \lfloor \frac{3 \ell}{2} \rfloor + \ell \leq 4 \ell$, so $\xsc(G_{k,\ell}) - \isc(G_{k,\ell}) = \ell - 2 > \frac{1}{18}k \ell$ since $\ell \geq 3$.  Thus the theorem holds when $k = 2$.

The next two lemmas complete the proof of Theorem \ref{thm:grid-diff}.  We begin by giving a lower bound on the sum choice number of grids.

\begin{lma}\label{lma:grid-sc}
Let $G_{k,\ell}$ denote the $k \times \ell$ grid, and suppose that $ \ell \geq k \geq 3$.  Then $\xsc(G_{k,\ell}) \geq \frac{23}{9}k \ell - \frac{2}{9}k \ell - k$.
\end{lma}
\begin{proof}
For any pair of non-negative integers $r$ and $s$, we define $H_{r,s}$ to be the graph consisting of the disjoint union of $r$ copies of $P_2 \square P_{\ell}$ and $s$ copies of $P_3 \square P_{\ell}$.  Now observe that for any positive integer $k \geq 3$, there exists $r \in \{0,1,2\}$ and a non-negative integer $s$ such that $G_{k,\ell}$ contains $H_{r,s}$ as a spanning subgraph.  Thus we know that $\xsc(G_{k,\ell}) \geq \xsc(H_{r,s})$, so it suffices to prove that the lower bound holds for $\xsc(H_{r,s})$ whenever $r \in \{0,1,2\}$ and $s \in \mathbb{Z}_0^+$ satisfy $2r + 3s = k$.

Recall (Theorem \ref{thm:p3pn}) that $\xsc(P_3 \square P_{\ell}) = 8 \ell - 3 - \lfloor \frac{\ell}{3} \rfloor \geq 8 \ell - 3 - \frac{\ell}{3}$, and that (as $P_2 \square P_{\ell}$ is sc-greedy) $\xsc(P_2 \square P_{\ell}) = 5 \ell - 2$.  Thus we see that
\begin{align*}
\xsc(H_{r,s}) & = r \left( \xsc(P_2 \square P_{\ell}) \right) + s \left( \xsc(P_3 \square P_{\ell}) \right) \\
              & \geq r (5 \ell - 2) + s \left(8 \ell - 3 - \frac{\ell}{3} \right) \\
              & = 5 r \ell - 2 r + 8 \ell s - 3 s - \frac{1}{3} \ell s \\
              & = \frac{23}{3} \ell s + 5 \ell r - 3 s - 2 r.
\end{align*}
There are now three cases to consider, depending on the value of $r$.

If $r = 0$ then $s = \frac{k}{3}$, and we have
\begin{align*}
    \xsc(H_{r,s}) & = \frac{23}{3} \ell \frac{k}{3}  - 3 \frac{k}{3} \\
                  & = \frac{23}{9} k \ell - k \\
                  & > \frac{23}{9}k \ell - \frac{2}{9}k \ell - k,
\end{align*}
as required.

Next suppose that $r = 1$, so $s = \frac{k-2}{3}$.  In this case we have
\begin{align*}
    \xsc(H_{r,s}) & = \frac{23}{3} \ell \frac{k-2}{3} + 5 \ell - 3 \frac{k-2}{3} - 2 \\
                  & = \frac{23}{9} k \ell - \frac{46}{9} \ell + 5 \ell - k + 2 - 2 \\
                  & = \frac{23}{9} k \ell - \frac{1}{9} \ell - k \\
                  & > \frac{23}{9}k \ell - \frac{2}{9}k \ell - k,
\end{align*}
as required.

Finally, suppose that $r = 2$ and hence $s = \frac{k-4}{3}$.  In this case,
\begin{align*}
    \xsc(H_{r,s}) & = \frac{23}{3} \ell \frac{k-4}{3} + 10 \ell - 3 \frac{k-4}{3} - 4 \\
                  & = \frac{23}{9} k \ell - \frac{92}{9} \ell + 10 \ell - k + 4 - 4 \\
                  & = \frac{23}{9} k \ell - \frac{2}{9} \ell - k,
\end{align*}
completing the proof.
\end{proof}

To complete the proof of Theorem \ref{thm:grid-diff}, we now derive an upper bound on the interactive sum choice number of grids by applying our knowledge about the interactive sum choice number of paths inductively.

\begin{lma}\label{lma:grid-ub}
Let $G_{k,\ell}$ denote the $k \times \ell$ grid, and assume that $\ell \geq k \geq 3$.  Then $\isc(G_{k,\ell}) \leq \frac{5}{2}\ell k - \frac{\ell}{2} - k$.
\end{lma}
\begin{proof}
There are three different cases to consider, depending on the parities of $\ell$ and $k$; in each case we will decompose $G_{k, \ell}$ into a number of paths and apply Lemma \ref{lma:del-subg} together with the fact that $\isc(P_n) = \left\lfloor \frac{3n}{2} \right\rfloor$.

For the first case, suppose that $\ell$ is odd.  In this case we consider decomposing $G_{k,\ell}$ into $k$ paths of length $\ell$ (each corresponding to one row of the grid), with $\ell$ edges between each pair of consecutive paths.  Applying Lemma \ref{lma:del-subg} repeatedly, we see that
\begin{align*}
    \isc(G_{k,\ell}) & \leq k \lfloor \frac{3 \ell}{2} \rfloor + (k-1)\ell \\
                     & = k \left(\ell + \frac{\ell - 1}{2}\right) + (k-1)\ell \\
                     & \qquad \qquad \mbox{since $\ell$ is odd} \\
                     & = \frac{5}{2} k\ell - \frac{k}{2} - \ell \\
                     & \leq \frac{5}{2} k \ell - \frac{\ell}{2} - k,
\end{align*}
since $\ell \geq k$.

For the second case, suppose that $k$ is odd.  In this case we consider decomposing $G_{k,\ell}$ into $\ell$ paths of length $k$, then by the same reasoning as in the first case we have once again that
$$ \isc(G_{k,\ell}) \leq \frac{5}{2} k \ell - \frac{\ell}{2} - k.$$

For the third and final case, suppose that $k$ and $\ell$ are both even.  In this case we decompose the graph into $k$ ``L-shaped'' paths, as illustrated in Figure \ref{fig:L-paths}.  The longest of these paths has length $k + \ell - 1$ and the shortest has length $\ell - k + 1$, with consecutive paths differing in length by $2$.

\begin{figure}
    \centering
    \includegraphics{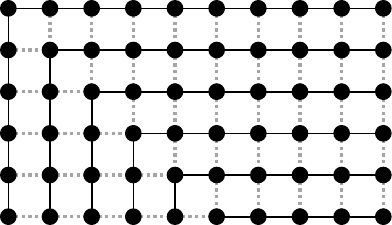}
    \caption{Decomposing $G_{k,\ell}$ into $k$ paths of odd length when $k$ and $\ell$ are both even.}
    \label{fig:L-paths}
\end{figure}

As before, we apply Lemma \ref{lma:del-subg} repeatedly to see that
$$\isc(G_{k,\ell}) \leq \sum_{i = 1}^k \isc(P_{\ell - k + 2i - 1}) + a,$$
where $a$ is the number of edges of $G_{k,\ell}$ not covered by our collection of paths.  Since the collection of paths forms a spanning forest for $G_{k,\ell}$ with $k$ components, it follows that the number of edges \emph{not} covered by this collection is exactly
$$|E(G_{k,\ell})| - (k \ell - k) = k(\ell-1) + \ell(k-1) - k \ell + k = k \ell - \ell.$$
Thus
\begin{align*}
    \isc(G_{k,\ell}) & \leq \sum_{i = 1}^k \left \lfloor \frac{3(\ell - k + 2i - 1)}{2} \right \rfloor + k \ell - \ell \\
                     & = \sum_{i = 1}^k \left( \frac{3}{2}\ell - \frac{3}{2}k + 3i - 2 \right) + k \ell - \ell \\
                     & \qquad \qquad \mbox{since $k$, $\ell$ even} \\
                     & = \frac{3}{2} k \ell - \frac{3}{2}k^2 - 2k + k \ell - \ell + 3\sum_{i=1}^k i \\
                     & = \frac{5}{2} k \ell - \frac{3}{2}k^2 - 2k - \ell + \frac{3}{2}k(k+1) \\
                     & = \frac{5}{2}k \ell - \ell - \frac{k}{2} \\
                     & \leq \frac{5}{2} k \ell - \frac{\ell}{2} - k,
\end{align*}
as required.
\end{proof}

\section{Graphs that may or may not be sc-greedy}
\label{sec:maybe-greedy}

In this section we see that we can prove Conjecture \ref{conj:all-non-complete} for certain graph classes where very little is known about the sum choice number, in particular when it is not even known whether the graphs in question are sc-greedy.

A \emph{good $2$-degenerate} graph is a graph that admits an ordering $\mathcal{O}=(v_1,\ldots,v_n)$ of the vertices so that each vertex has at most $2$ neighbours later in the order and only has two if it belongs to a cycle whose vertices lie later in the order. Let $q(G)$ be the number of vertices that belong to a cycle whose vertices lie later in the order.  Cacti (graphs in which no two cycles share an edge) are examples of good $2$-degenerate graphs.

\begin{thm}\label{thm:2deg}
Any good $2$-degenerate graph $G$ on $n$ vertices satisfies $\isc(G) \leq \frac{3(n + q(G))}2$.
\end{thm}

\begin{proof}
By contradiction. Take a minimal counter-example $G$, consider $n$ its number of vertices and let $q=q(G)$. By Theorem~\ref{thm:trees}, we have $q(G) \geq 1$. Consider an ordering $\mathcal{O}=(v_1,\ldots,v_n)$ of the vertices so that each vertex has at most $2$ neighbours later in the order and only has two if it belongs to a cycle whose vertices lie later in the order. For every $i$, let $G_i$ be the subgraph induced in $G$ by the vertices $(v_i,v_{i+1},\ldots,v_n)$. Let $\alpha$ be a function of first choices such that $G$ has no $\alpha$-strategy of length $q+\frac{n+q}2$. Re-using the arguments for trees, every vertex in $G$ is adjacent to at most one leaf, and no vertex of degree $2$ is adjacent to a leaf. In particular, every cut-edge $e$ is either adjacent to a vertex of degree $1$ or separates $G$ in two parts each of which contains at least one cycle. Let $i$ be the smallest integer such that $v_i$ belongs to a cycle in $G_i$. By choice of our ordering, the integer $i$ is also the smallest such that $v_i$ is of degree $2$ in $G_i$. Let $u$ and $w$ be the two neighbours of $v_i$ in $G_i$. Note that since no vertex not in $G_i$ belongs to a cycle, every edge incident to $v_i$ that is not $uv_i$ nor $wv_i$ is a cut-edge, and is therefore incident to a leaf. Remember that no vertex is adjacent to more than one leaf, and let $x$ be the possible leaf $v_i$ is adjacent to. We have $N(v_i) \subseteq \{u,w,x\}$.

We first assume that $x$ exists. We note that $G'=G-\{v_i,x\}$ admits a strategy of length $\frac{3((n-2) + (q(G)-1))}2=\frac{3(n + q(G))}2-\frac92$, by minimality of $G$. If $\alpha(v_i)\neq \alpha(x)$, we run the strategy on $G'$ using Lemma~\ref{lma:notc} to enforce that $u$ and $w$ should receive a colour distinct from $\alpha(v_i)$: this costs only 2 extra requests. In total, the strategy ran in at most $\frac{3(n + q(G))}2-\frac92+2+2$ rounds, a contradiction with the choice of $G$. If $\alpha(v_i)=\alpha(x)$, we run the strategy on $G'$, and consider two cases depending on whether one of $u$ and $w$ is coloured with $\alpha(v_i)$: if so, we request two new colours for $v_i$. If not, we request one new colour for $x$. All in all, the strategy ran in at most $\frac{3(n + q(G))}2-\frac92+2+2$ rounds, a contradiction with the choice of $G$.

Let us now deal with the case where $x$ does not exist. We note that $G'=G-\{v_i\}$ admits a strategy of length $\frac{3((n-1) + (q(G)-1))}2=\frac{3(n + q(G))}2-3$, by minimality of $G$. We run the strategy on $G'$, then request two new colours on $v_i$: since $d(v_i)=2$ this is enough to guarantee the proper colouring of $G'$ can be extended to $G$. the strategy ran in at most $\frac{3(n + q(G))}2-3+1+2$ rounds, a contradiction with the choice of $G$.
\end{proof}

Combining this result with Lemma \ref{lma:non-treesarelosers}, we obtain the following immediate corollary.

\begin{cor}
If $G$ is a good $2$-degenerate graph with $q(G) < \frac{n}{3}$, then $\isc(G) < \xsc(G)$.
\end{cor}

\section{Conclusions and open problems}

We have introduced the interactive sum choice number of graphs, a variation of the sum choice number in which we are able to exploit partial information about the contents of colour lists in order to inform our strategy.  We demonstrated that in many cases this additional information allows us to guarantee a proper list colouring when the sum of list lengths over all the vertices is strictly smaller than the sum choice number of the graph, and for several families of graphs we were in fact able to prove the existence of a large gap between the sum choice number and the interactive sum choice number.

As is often the case when a new problem is introduced, this paper raises more questions than it solves. The key open question arising from this work is to prove Conjecture \ref{conj:all-non-complete}, namely that if $G$ is not a complete graph then $\isc(G) < \xsc(G)$; a first step would be to attempt to prove the conjecture for further graph classes, for example $k$-degenerate graphs, chordal graphs, planar graphs, cographs or graphs of bounded treewidth. Since graphs with high degeneracy are known to have high choice number~\cite{Nogaaaa} with a proof that only really uses arguments around one arbitrary vertex, it might be worth trying to prove similarly that they have (very) high sum choice number. In turn, that would be a step towards Conjecture~\ref{conj:all-non-complete} for graphs with high degeneracy.

It would also be interesting to investigate further just how much these two quantities can differ; in particular, the upper bounds on the interactive sum choice number that we have obtained for unbalanced bipartite graphs and grids are unlikely to be tight, so it seems natural to seek better bounds for these graph classes.  A natural next step is to attempt to find further classes of graphs for which the difference between the sum choice number and the interactive sum choice number is a growing function of the number of vertices. On the other hand, what can we say about the structure of graphs for which the difference between the sum choice number and interactive sum choice number is bounded by some constant independent of the number of vertices?

In addressing any of these questions, it would be extremely helpful to understand how to use cut-edges, cut-vertices, modules, joins, and similar decompositions of graphs. Also, tools to prove lower bounds on the interactive sum choice number are sorely missing.

\section*{Acknowledgements}

Kitty Meeks is supported by a Royal Society of Edinburgh Personal Research Fellowship, funded by the Scottish Government.   Both authors would like to thank Ond\v{r}ei Splichal and Jakub L\"{o}wit for pointing out an error in a preliminary version of this paper.

\bibliographystyle{plain}
\bibliography{sumchoice}

\begin{thebibliography}{10}

\bibitem{Nogaaaa}
Noga Alon.
\newblock Degrees and choice numbers.
\newblock {\em Random Structures Algorithms}, 16(4):364--368, 2000.

\bibitem{berliner06}
Adam Berliner, Ulrike Bostelmann, A.~Richard Brualdi, and Louis Deaett.
\newblock Sum list coloring graphs.
\newblock {\em Graphs and Combinatorics}, 22(2):173--183, 2006.

\bibitem{carraher15}
James~M. Carraher, Thomas Mahoney, Gregory~J. Puleo, and Douglas~B. West.
\newblock Sum-paintability of generalized theta-graphs.
\newblock {\em Graphs and Combinatorics}, 31(5):1325--1334, 2015.

\bibitem{Diestel}
R.~Diestel.
\newblock {\em Graph Theory}, volume 173 of {\em Graduate Texts in
  Mathematics}.
\newblock Springer-Verlag, Heidelberg, third edition, 2005.

\bibitem{furedi16}
Zoltán Füredi and Ida Kantor.
\newblock List colorings with distinct list sizes, the case of complete
  bipartite graphs.
\newblock {\em Journal of Graph Theory}, 82(2):218--227, 2016.

\bibitem{heinold-thesis}
Brian Heinold.
\newblock {\em Sum list coloring and choosability}.
\newblock PhD thesis, Lehigh University, 2006.

\bibitem{heinold12}
Brian Heinold.
\newblock The sum choice number of ${P}_3 \square {P}_n$.
\newblock {\em Discrete Applied Mathematics}, 160(7–8):1126 -- 1136, 2012.

\bibitem{isaak02}
Garth Isaak.
\newblock Sum list coloring 2 by $n$ arrays.
\newblock {\em Electronic Journal of Combinatorics}, 9:Note 8, 2002.

\bibitem{isaak04}
Garth Isaak.
\newblock Sum list coloring block graphs.
\newblock {\em Graphs and Combinatorics}, 20(4):499--506, 2004.

\bibitem{wheels}
A.~Kemnitz, M.~Marangio, and M.~Voigt.
\newblock Sum list colorings of wheels.
\newblock {\em Graphs and Combinatorics}, 31:1905–--1913, 2015.

\bibitem{kemnitz17}
Arnfried Kemnitz, Massimiliano Marangio, and Margit Voigt.
\newblock Bounds for the sum choice number.
\newblock {\em Electronic Notes in Discrete Mathematics}, 63:49 -- 58, 2017.
\newblock International Conference on Current Trends in Graph Theory and
  Computation.

\bibitem{lastrina13}
Michelle Lastrina and Michael Young.
\newblock Sum list coloring, the sum choice number, and sc-greedy graphs.
\newblock arXiv:1305.1962 [math.CO], 2013.

\bibitem{mahoney15slowcol}
Thomas Mahoney, Gregory~J Puleo, and Douglas~B West.
\newblock Online paintability: The slow-coloring game.
\newblock {\em arXiv preprint arXiv:1507.06513}, 2015.

\bibitem{mahoney15}
Thomas Mahoney, Charles Tomlinson, and Jennifer~I. Wise.
\newblock Families of online sum-choice-greedy graphs.
\newblock {\em Graphs and Combinatorics}, 31(6):2309--2317, 2015.

\bibitem{puleo19}
Gregory~J. Puleo and Douglas~B. West.
\newblock Online sum-paintability: Slow-coloring of trees.
\newblock {\em Discrete Applied Mathematics}, 262:158 -- 168, 2019.

\end{thebibliography}

\end{document}